\documentclass[journal]{IEEEtran}
\usepackage{balance}
\usepackage{color}
\usepackage{cite}
\usepackage{graphicx}
\usepackage{amsmath}
\usepackage{amsthm}
\usepackage{amssymb}
\usepackage{balance}
\usepackage{epstopdf}
\usepackage{subfig}
\usepackage{float}
\usepackage{epsfig}
\newtheorem{innercustomthm}{Theorem}
\newenvironment{customthm}[1]
  {\renewcommand\theinnercustomthm{#1}\innercustomthm}
  {\endinnercustomthm}
\usepackage[T1]{fontenc}
\usepackage{amsfonts}
\usepackage{listings}
\usepackage{flushend}
\usepackage{fancybox}
\usepackage{epsfig}

\newtheorem{theorem}{Theorem}
\newtheorem{lemma}[theorem]{Lemma}

\newcommand{\imi}{\mathsf{i}}
\newcommand{\imj}{\mathsf{j}}
\newcommand{\imk}{\mathsf{k}}
\begin{document}
\title{A quaternion-based approach to construct quaternary periodic complementary pairs}
\author{{\IEEEauthorblockN{Nitin Jonathan Myers, {\it Student Member, IEEE}, and 
Robert W. Heath Jr., {\it Fellow, IEEE}. }}
\thanks{N. J. Myers (nitinjmyers@utexas.edu) and R. W. Heath Jr. (rheath@utexas.edu) are with the  Wireless Networking and Communications  Group, The University of Texas at Austin, Austin, TX 78712 USA. This material is based upon work supported by the National Science Foundation under Grant numbers NSF-CNS-1731658 and NSF ECCS-1711702. }
}
\maketitle

\begin{abstract}
Two arrays form a periodic complementary pair if the sum of their periodic autocorrelations is a delta function. Finding such pairs, however, is challenging for large arrays whose entries are constrained to a small alphabet. One such alphabet is the quaternary set which contains the complex fourth roots of unity. In this paper, we propose a technique to construct periodic complementary pairs defined over the quaternary set using perfect quaternion arrays. We show how new pairs of quaternary sequences, matrices, and four-dimensional arrays that satisfy a periodic complementary property can be constructed with our method. 
\end{abstract}

\section{Introduction}
A periodic complementary pair (PCP) is a collection of two arrays whose periodic autocorrelation sums up to a delta function. An array can refer to a sequence, matrix, or a tensor. For example, a sequence or vector of length $M$ is a one-dimensional array of size $M$, while an $M \times N$ matrix is a two-dimensional array of size $M \times N$. PCPs are arrays with special periodic autocorrelation properties that find applications in coded aperture imaging \cite{codedapert}, communications \cite{commun_seq}, and radar \cite{radar_app}. PCPs are different than other complementary sequence constructions like Golay pairs \cite{golay_seq}, where the notion of complementary applies to aperiodic correlations. 
\par Prior work has considered the design of PCPs over a binary alphabet \cite{pcp_bomer,pcp_dokovic,pcp_twoval}. The constructions in \cite{pcp_bomer,pcp_dokovic,pcp_twoval} are for sequences, i.e., $1\mathrm{D}$ arrays, over $\{1,-1\}$ or $\{1,e^{\mathsf{i \theta}}\}$ where $\mathsf{i}=\sqrt{-1}$ is the standard unit imaginary number. A class of two-dimensional arrays in $\{-1,1\}$ which are PCPs was derived in \cite{pcp_2D}. PCPs over binary alphabets, however, do not exist for every array size. For example, it was shown in \cite{no_exist_18} that binary PCP sequences of length $18$ do not exist. Relaxing the binary constraint on the alphabet to a quaternary one provides additional flexibility to design new PCPs.
\par Sequences over the quaternary alphabet $\mathcal{C}=\{1, \mathsf{i},-1, -\mathsf{i} \}$ that form PCPs were proposed in \cite{qpcp_1,qpcp_2,qpcp_3}. The quaternary PCPs in \cite{qpcp_1,qpcp_2,qpcp_3} were generated using the Gray method, inverse Gray method with interleaving, or the product method. An extensive survey on the length of sequences for which a quaternary PCP exists can be found in \cite{qpcp_1}. The literature on quaternary PCPs, however, is limited to sequences. The existence and construction of two-dimensional or multi-dimensional quaternary PCPs has not been studied to the best of our knowledge. One application of $2\mathrm{D}$ quaternary PCPs is in dual polarized planar antenna arrays \cite{petersson2019power} equipped with two-bit phase shifters. Quaternary PCPs when applied to such systems result in quasi-omnidirectional beams which are useful for initial access in millimeter wave systems. 
\par In this paper, we construct new one-, two-, and four-dimensional quaternary PCPs by leveraging perfect quaternion arrays \cite{pqa_1}. It is important to note that the entries of the PCPs constructed in this paper contain elements in $\mathcal{C}$ which are complex numbers. The construction, however, is derived using quaternion algebra which is different from standard algebra over complex numbers \cite{q_algebra_1}. In Section \ref{sec:connection}, we exploit the fact that the right periodic autocorrelation of a perfect quaternion array (PQA) is a delta function to construct PCPs. Our construction allows decomposing any PQA into a PCP. We use this construction in Section \ref{sec:construct} to show that PQAs over the basic unit quaternions \cite{pqa_1,pqa_2, pqa_3, pqa_4} can be mapped to quaternary PCPs of the same size. In Section \ref{sec:leftperiodic}, we construct a different class of PCPs by leveraging the properties of the Fourier transform and the left periodic autocorrelation of PQAs. A Matlab implementation of the quaternary PCPs derived in this paper is available on our GitHub page \cite{PCP_git}. 
\par \par Notation: $\mathbf{A}$ is a matrix, $\mathbf{a}$ is a column vector and $a, A$ denote scalars. $\mathbf{A}^{\ast}$ and $a^{\ast}$ denote the complex conjugate of $\mathbf{A}$ and $a$. We use $\mathbf{1}$ to denote an all-ones matrix. $A\left(k,\ell\right)$ denotes the entry of $\mathbf{A}$ in the $k^{\mathrm{th}}$ row and the ${\ell}^{\mathrm{th}}$ column. The indices $k$ and $\ell$ start from $0$. $\Vert \mathbf{A} \Vert _{\mathrm{F}}$ is the Frobenius norm of $\mathbf{A}$. $\mathbb{R}$, $\mathbb{C}$, and $\mathbb{Q}$ denote the set of real, complex, and quaternion numbers. $\langle \cdot \rangle_N$ denotes the modulo $N$ operation. We define the flipped version of an $M\times N$ matrix $\mathbf{A}$ as $\mathbf{A}_{\mathrm{flip}}$ where $A_{\mathrm{flip}}(m,n)=A(\langle-m\rangle_M,\langle-n\rangle_N)$.
\section{Preliminaries}
\par Quaternions are a generalization of complex numbers which have one real component and three imaginary components \cite{q_algebra_1}. We define $\imi$, $\imj$, and $\imk$ as the fundamental quaternion units. These units satisfy
\begin{align}
\nonumber
&\imi^2=-1,&& \imj^2=-1, && \imk^2=-1,\\
\nonumber
&\imi \imj =\imk, && \imj \imk =\imi, && \imk \imi =\imj, \\
\label{eq:prop_quatn}
&\imj \imi  =-\imk, && \imk \imj =-\imi, && \imi \imk  =-\imj.
\end{align}
Any quaternion $q \in \mathbb{Q}$ can be expressed as \cite{q_algebra_1}
\begin{equation}
q=q_1+q_2 \imi + q_3 \imj +q_4 \imk, 
\end{equation}
where $q_1, \, q_2, \, q_3, \,q_4 \in \mathbb{R}$. Complex numbers are a special instance of quaternions, i.e., $\mathbb{C}=\{q\in \mathbb{Q}:q_3=0, q_4=0\}$. We use the properties in \eqref{eq:prop_quatn} and the distributive law to express the quaternion $q$ in terms of two complex numbers as  
\begin{equation}
\label{eq:quat_split}
q=\underbrace{(q_1+q_2 \imi)}_{q_{\mathrm{h}}} + \underbrace{(q_3 +q_4 \imi)}_{q_{\mathrm{v}}} \imj.
\end{equation} 
The quaternion $q$ can be written as $q=q_{\mathrm{h}}+q_{\mathrm{v}} \imj$.
\par We now discuss basic operations over quaternions. The product of two quaternions $p$ and $q$ is \cite{q_algebra_1}
\begin{align}
\nonumber
p q &=(p_1q_1-p_2q_2-p_3q_3-p_4q_4)\\
\nonumber
&+ (q_1p_2 + p_1q_2 - q_3p_4 + p_3q_4)\imi\\
\nonumber
&+ (q_1p_3 + p_1q_3 + q_2p_4 - p_2q_4)\imj \\
\label{eq:prod_quatn}
&+ (q_1p_4 + p_1q_4 - q_2p_3 + p_2q_3)\imk.
\end{align}
Multiplication over quaternions is non-commutative. For example, it can be observed from \eqref{eq:prop_quatn} that $\imi \imj \neq \imj \imi $. It is important to note, however, that $\alpha q=q \alpha$ for $ \alpha \in \mathbb{R}$ and $q \in \mathbb{Q}$. The complex conjugate of $q$ is \cite{q_algebra_1} 
\begin{equation}
\label{eq:conj_quatn}
q^{\ast}=q_1-q_2 \imi - q_3 \imj -q_4 \imk.
\end{equation}
The conjugates corresponding to the product and sum of two quaternions can be expressed as
\begin{align}
\label{eq:prod_conj}
(pq)^{\ast}&=q^{\ast}p^{\ast} \, \mathrm{and}\\
\label{eq:sum_conj}
(p+q)^{\ast}&=p^{\ast}+q^{\ast}.
\end{align}
The properties in \eqref{eq:prod_quatn}$-$\eqref{eq:sum_conj} naturally extend to matrices.   
\par Now, we define matrices over quaternions and the periodic autocorrelation of a quaternion matrix. Consider a quaternion matrix $\mathbf{A}\in \mathbb{Q}^{M \times N}$. Similar to the representation in \eqref{eq:quat_split}, we use $\mathbf{A}_{\mathrm{h}} \in \mathbb{C}^{M \times N}$ and $\mathbf{A}_{\mathrm{v}} \in \mathbb{C}^{M \times N}$ to denote the complex components of $\mathbf{A}$ such that 
\begin{equation}
\label{eq:quat_mat_dec}
\mathbf{A}= \mathbf{A}_{\mathrm{h}}+\mathbf{A}_{\mathrm{v}}\imj.
\end{equation}
The non-commutative nature of quaternion multiplication leads to a different left and right periodic correlation \cite{pqa_2}. In this paper, we focus on the right periodic correlation to derive PCPs. We also briefly derive a different class of PCPs using the left periodic correlation in Section \ref{sec:leftperiodic}. For two matrices $\mathbf{X}, \mathbf{Y} \in \mathbb{Q}^{M \times N}$, we define the conjugate-free periodic cross correlation as the matrix $\mathbf{X}\star \mathbf{Y} \in \mathbb{Q}^{M \times N}$. The $(m,n)^{\mathrm{th}}$ entry of $\mathbf{X}\star \mathbf{Y}$ is
\begin{equation}
\label{eq:stardefn}
(\mathbf{X}\star \mathbf{Y})_{m,n}=\sum^{M-1}_{k=0} \sum^{N-1}_{\ell=0} X(k,\ell)Y(\langle k+m \rangle_M,\langle \ell+n \rangle _N).
\end{equation}
We use $\mathbf{R}_{\mathbf{A}}$ to denote the 2D-periodic autocorrelation of $\mathbf{A}$. The $(m,n)^{\mathrm{th}}$ entry of $\mathbf{R}_{\mathbf{A}} \in \mathbb{Q}^{M \times N}$ is \cite{pqa_2}
\begin{equation}
\label{eq:ACR_quatn}
R_{\mathbf{A}}(m,n)=\sum^{M-1}_{k=0} \sum^{N-1}_{\ell=0} A(k,\ell)A^{\ast}(\langle k+m \rangle_M,\langle \ell+n \rangle _N).
\end{equation}
It can be observed that $\mathbf{R}_{\mathbf{A}}=\mathbf{A} \star \mathbf{A}^{\ast}$. For the special case when $\mathbf{A}$ is a complex matrix, i.e., when $\mathbf{A}$ does not have $\imj$ and $\imk$ components,  the right periodic autocorrelation $\mathbf{R}_{\mathbf{A}}$ is the common 2D-periodic autocorrelation of $\mathbf{A}$.
\par We now define periodic complementary matrix pairs over the complex numbers. We use $\boldsymbol{\delta}$ to denote a unit delta matrix of size $M \times N$, i.e., $\delta(0,0)=1$ and $\delta(m,n)=0\, \forall \, (m, n) \neq (0,0)$. Two matrices $\mathbf{X} \in \mathbb{C}^{M \times N}$ and $\mathbf{Y}  \in \mathbb{C}^{M \times N}$ form a PCP if \cite{pcp_2D} 
\begin{equation}
\label{eq:PCP_cmp}
R_{\mathbf{X}}(m,n)+R_{\mathbf{Y}}(m,n)=0\,\, \forall (m,n)\neq (0,0).
\end{equation} 
Equivalently, $\mathbf{R}_{\mathbf{X}}+\mathbf{R}_{\mathbf{Y}}=2MN \boldsymbol{\delta}$ for a PCP with $\Vert \mathbf{X}\Vert_{\mathrm{F}}=\sqrt{MN}$ and $\Vert \mathbf{Y}\Vert_{\mathrm{F}}=\sqrt{MN}$. A trivial PCP is $\mathbf{X}=\sqrt{MN} \boldsymbol{\delta}$ and $\mathbf{Y}=\sqrt{MN} \boldsymbol{\delta}$. Finding PCPs with entries in $\{1, -1\}$ or $\mathcal{C}$, however, is challenging when the size of the matrices is large. In this paper, we propose a technique to construct new PCPs whose entries are in $\mathcal{C}$. 
\section{Connection between perfect quaternion arrays and complex periodic complementary pairs} \label{sec:connection}
\par In this section, we show that every perfect quaternion array (PQA) can be decomposed into a PCP with complex entries in $\mathbb{C}$. A matrix $\mathbf{A} \in \mathbb{Q}^{M\times N}$ is a PQA if its right periodic autocorrelation is a delta function, i.e., 
\begin{equation}
\label{eq:PQA_def}
R_{\mathbf{A}}(m,n)=0\,\, \forall (m,n)\neq (0,0).
\end{equation}
This property can be expressed as $\mathbf{R}_{\mathbf{A}}=MN\boldsymbol{\delta}$ when $\Vert \mathbf{A}\Vert_{\mathrm{F}}=\sqrt{MN}$. An example of a $2 \times 2$ PQA is \cite{pqa_4}
\begin{equation}
\label{eq:example_PQA}
\mathbf{D}=\begin{pmatrix}
1 & \imi \\
\imj & \imk  
\end{pmatrix}.
\end{equation}
Quaternion matrices that are PQAs were investigated in \cite{pqa_1,pqa_2,pqa_3,pqa_4}. To the best of our knowledge, prior work has not studied the connection between PQAs and PCPs.
\par We first express the autocorrelation $\mathbf{R}_{\mathbf{A}}$ as a function of the complex matrices $\mathbf{A}_{\mathrm{h}}$ and $\mathbf{A}_{\mathrm{v}}$ in Lemma 1. The result in Lemma 1 is then used in Theorem \ref{th:main} to derive PCPs.
 \begin{lemma}
 \label{lem:decomp}
 For any quaternion matrix $\mathbf{A}$, 
 \begin{equation}
 \mathbf{R}_{\mathbf{A}}= \mathbf{R}_{\mathbf{A}_{\mathrm{h}}}+\mathbf{R}_{\mathbf{A}_{\mathrm{v}}}+ [\mathbf{A}_{\mathrm{v}}\star\mathbf{A}_{\mathrm{h}}- \mathbf{A}_{\mathrm{h}}\star\mathbf{A}_{\mathrm{v}}]\imj.
 \end{equation}
 \end{lemma}
 \begin{proof}
We use the complex decomposition in \eqref{eq:quat_mat_dec} to write $\mathbf{R}_{\mathbf{A}}=(\mathbf{A}_{\mathrm{h}}+\mathbf{A}_{\mathrm{v}}\imj) \star (\mathbf{A}_{\mathrm{h}}+\mathbf{A}_{\mathrm{v}}\imj)^{\ast}$. Using the distributive law and \eqref{eq:prod_conj}, the autocorrelation can be simplified to
\begin{equation}
\label{eq:RA_exp_1}
\mathbf{R}_{\mathbf{A}}=\mathbf{A}_{\mathrm{h}} \star \mathbf{A}^{\ast}_{\mathrm{h}} + \mathbf{A}_{\mathrm{v}}\imj \star \imj^{\ast} \mathbf{A}^{\ast}_{\mathrm{v}}+\mathbf{A}_{\mathrm{h}}\star\imj^{\ast}\mathbf{A}^{\ast}_{\mathrm{v}}+ \mathbf{A}_{\mathrm{v}}\imj \star  \mathbf{A}^{\ast}_{\mathrm{h}}.
\end{equation}
The first summand in \eqref{eq:RA_exp_1} is $\mathbf{R}_{\mathbf{A}_{\mathrm{h}}}$. The second summand in \eqref{eq:RA_exp_1} is simplified using \eqref{eq:stardefn} and the property that $\imj \imj^{\ast}=1$. The simplification results in $ \mathbf{A}_{\mathrm{v}}\imj \star \imj^{\ast} \mathbf{A}^{\ast}_{\mathrm{v}}= \mathbf{R}_{\mathbf{A}_{\mathrm{v}}}$. Therefore, $\mathbf{R}_{\mathbf{A}}=\mathbf{R}_{\mathbf{A}_{\mathrm{h}}}+\mathbf{R}_{\mathbf{A}_{\mathrm{v}}}+\mathbf{A}_{\mathrm{h}}\star\imj^{\ast}\mathbf{A}^{\ast}_{\mathrm{v}}+ \mathbf{A}_{\mathrm{v}}\imj \star  \mathbf{A}^{\ast}_{\mathrm{h}}$.
\par We now show that the sum of the third and the fourth summands in \eqref{eq:RA_exp_1} is $[\mathbf{A}_{\mathrm{v}}\star\mathbf{A}_{\mathrm{h}}- \mathbf{A}_{\mathrm{h}}\star\mathbf{A}_{\mathrm{v}}]\imj$. We define $\alpha_{m,n}$ as the $(m,n)^{\mathrm{th}}$ element of $\mathbf{A}_{\mathrm{h}}\star\imj^{\ast}\mathbf{A}^{\ast}_{\mathrm{v}}+ \mathbf{A}_{\mathrm{v}}\imj \star  \mathbf{A}_{\mathrm{h}}$. Then,
\begin{align}
\nonumber
\alpha_{m,n}=\sum^{M-1}_{k=0}\sum^{M-1}_{k=0} \big[&A_{\mathrm{h}}(k,\ell)\imj^{\ast}A_{\mathrm{v}}^{\ast}(\langle k+m \rangle_M,\langle \ell+n \rangle _N)\\
\label{eq:alpha_mn}
&+A_{\mathrm{v}}(k,\ell)\imj A_{\mathrm{h}}^{\ast}(\langle k+m \rangle_M,\langle \ell+n \rangle _N) \big].
\end{align} 
To simplify $\alpha_{m,n}$, we use the property that $x \imj y = x y^{\ast} \imj$ and $x \imj^{\ast} y = -x y^{\ast} \imj$ for $x \in \mathbb{C}$ and $y\in \mathbb{C}$ [Proof in the Appendix]. As the entries of $\mathbf{A}_{\mathrm{h}}$ and $\mathbf{A}_{\mathrm{v}}$ are elements in $\mathbb{C}$, this property can be used in \eqref{eq:alpha_mn} to 
show that $\alpha_{m,n}$ is the $(m,n)^{\mathrm{th}}$ entry of $[\mathbf{A}_{\mathrm{v}}\star\mathbf{A}_{\mathrm{h}}- \mathbf{A}_{\mathrm{h}}\star\mathbf{A}_{\mathrm{v}}]\imj$.
 \end{proof}
 \begin{customthm}{1}\label{th:main}
The complex components $\mathbf{A}_{\mathrm{h}}$ and $\mathbf{A}_{\mathrm{v}}$ of a perfect quaternion array $\mathbf{A}$ form a PCP.
\end{customthm} 
\begin{proof}
When $\mathbf{A}$ is a perfect quaternion array, i.e., $\mathbf{R}_{\mathbf{A}}=MN \boldsymbol {\delta}$, the result in Lemma \ref{lem:decomp} leads to
\begin{equation}
\label{eq:equatePQA_zero}
\mathbf{R}_{\mathbf{A}_{\mathrm{h}}}+\mathbf{R}_{\mathbf{A}_{\mathrm{v}}}+ [\mathbf{A}_{\mathrm{v}}\star\mathbf{A}_{\mathrm{h}}- \mathbf{A}_{\mathrm{h}}\star\mathbf{A}_{\mathrm{v}}]\imj=MN \boldsymbol {\delta}.
\end{equation}
We interpret the quaternion matrix on the left hand side of \eqref{eq:equatePQA_zero} as a sum of $\mathbf{R}_{\mathbf{A}_{\mathrm{h}}}+\mathbf{R}_{\mathbf{A}_{\mathrm{v}}}$ and $[\mathbf{A}_{\mathrm{v}}\star\mathbf{A}_{\mathrm{h}}- \mathbf{A}_{\mathrm{h}}\star\mathbf{A}_{\mathrm{v}}]\imj$. As $\mathbf{A}_{\mathrm{h}}$ and $\mathbf{A}_{\mathrm{v}}$ are matrices in $\mathbb{C}^{M \times N}$, the first term $\mathbf{R}_{\mathbf{A}_{\mathrm{h}}}+\mathbf{R}_{\mathbf{A}_{\mathrm{v}}} \in \mathbb{C}^{M\times N}$ and does not have any $\imj$ and $\imk$ components. The second term, i.e., $[\mathbf{A}_{\mathrm{v}}\star\mathbf{A}_{\mathrm{h}}- \mathbf{A}_{\mathrm{h}}\star\mathbf{A}_{\mathrm{v}}]\imj$, is a multiplication of a matrix in $\mathbb{C}^{M \times N}$ with $\imj$. Such a matrix has zero real and zero $\imi$ components. The matrix on the right hand side, however, is purely real. Putting these observations together, it can be concluded that the equality in \eqref{eq:equatePQA_zero} holds only when 
\begin{align}
\label{eq:thmpcp}
\mathbf{R}_{\mathbf{A}_{\mathrm{h}}}+\mathbf{R}_{\mathbf{A}_{\mathrm{v}}}&=MN\boldsymbol{\delta},\, \mathrm{and}\\
\label{eq:corr_commut}
\mathbf{A}_{\mathrm{v}}\star\mathbf{A}_{\mathrm{h}}&= \mathbf{A}_{\mathrm{h}}\star\mathbf{A}_{\mathrm{v}}.
\end{align}
The result in Theorem \ref{th:main} follows from \eqref{eq:thmpcp}. 
\end{proof}
\par An important observation from Theorem \ref{th:main} is that the complex components of a PQA, i.e., $\mathbf{A}_{\mathrm{h}}$ and $\mathbf{A}_{\mathrm{v}}$, satisfy \eqref{eq:corr_commut} in addition to the PCP property in \eqref{eq:thmpcp}. Equivalently, our method results in PCPs which are commutative with conjugate-free periodic cross-correlation. We now explain an example to generate a PCP in $\mathbb{C}^{2 \times 2}$ from the PQA in \eqref{eq:example_PQA}. The matrix $\mathbf{D}$ in \eqref{eq:example_PQA} can be expressed as $\mathbf{D}_{\mathrm{h}}+\mathbf{D}_{\mathrm{v}} \imj$, where 
\begin{equation}
\label{eq:d_hv}
\mathbf{D}_{\mathrm{h}}=\begin{pmatrix}
1 & \imi \\
0 & 0  
\end{pmatrix} \, \mathrm{and} \, \,
\mathbf{D}_{\mathrm{v}}=\begin{pmatrix}
0 & 0 \\
1 & \imi  
\end{pmatrix}.
\end{equation}
It can be verified that $\mathbf{D}_{\mathrm{h}}$ and  $\mathbf{D}_{\mathrm{h}}$ satisfy the definition of a PCP in \eqref{eq:PCP_cmp}, equivalently \eqref{eq:thmpcp}. These matrices, however, contain entries which are not in the quaternary alphabet $\mathcal{C}$. In Section \ref{sec:construct}, we show how Theorem \ref{th:main} can still be used to construct quaternary PCPs from PQAs. 
\section{Construction of quaternary PCPs from PQAs} \label{sec:construct}
We define the basic unit quaternion alphabet as $\mathbb{H}=\{1,-1,\imi,-\imi, \imj, -\imj, \imk, -\imk\}$. PQAs with entries in $\mathbb{H}$ were constructed in \cite{pqa_1,pqa_2,pqa_3,pqa_4}. In this section, we show how to construct PCPs with entries in $\mathcal{C}$ from such PQAs. 
\par To generate quaternary PCPs, we first construct a matrix $\tilde{\mathbf{A}}=\mathbf{A}(1+\imj)$ where $\mathbf{A}$ is a PQA in $\mathbb{H}^{M \times N}$. The periodic autocorrelation of $\tilde{\mathbf{A}}$, i.e., $\mathbf{R}_{\tilde{\mathbf{A}}}$, is then $\mathbf{A}(1+\imj)\star (\mathbf{A}(1+\imj))^{\ast}$. The autocorrelation can be further simplified to $\mathbf{R}_{\tilde{\mathbf{A}}}=\mathbf{A}(1+\imj) \star (1+\imj)^{\ast} \mathbf{A}^{\ast}$ using \eqref{eq:prod_conj}. As $(1+\imj)(1+\imj)^{\ast}=2$, it can be shown that $\mathbf{R}_{\tilde{\mathbf{A}}}=2\mathbf{R}_{\mathbf{A}}$. Now, it follows from \eqref{eq:PQA_def} that $\tilde{\mathbf{A}}$ is a PQA whenever $\mathbf{A}$ is a PQA. 
\par We observe that the complex components of $\tilde{\mathbf{A}}$, i.e., $\tilde{\mathbf{A}}_{\mathrm{h}}$ and $\tilde{\mathbf{A}}_{\mathrm{v}}$ form a PCP using Theorem 1. These components can be expressed in terms of $\mathbf{A}_{\mathrm{h}}$ and $\mathbf{A}_{\mathrm{v}}$ as
\begin{align}
\nonumber
\tilde{\mathbf{A}}_{\mathrm{h}}&=\mathbf{A}_{\mathrm{h}}-\mathbf{A}_{\mathrm{v}}\,\, \mathrm{and}\\
\label{eq:tildes}
\tilde{\mathbf{A}}_{\mathrm{v}}&=\mathbf{A}_{\mathrm{h}}+\mathbf{A}_{\mathrm{v}}.
\end{align}
When $\mathbf{A}$ is a PQA in $\mathbb{H}^{M \times N}$, the entries of $\mathbf{A}_{\mathrm{h}}$ and $\mathbf{A}_{\mathrm{v}}$ are elements in $\mathcal{C} \cup \{0\}$. In addition, $\forall k, \ell$, $A_{\mathrm{h}}(k, \ell)=0$ whenever $A_{\mathrm{v}} (k,\ell)\neq0$ and vice versa. Putting these observations together, the entries of $\tilde{\mathbf{A}}_{\mathrm{h}}$ and $\tilde{\mathbf{A}}_{\mathrm{v}}$ are elements in $\{1, \imi, -1, -\imi\}$. Therefore, $\tilde{\mathbf{A}}_{\mathrm{h}}$ and $\tilde{\mathbf{A}}_{\mathrm{v}}$ form a quaternary PCP in $\mathcal{C}^{M \times N}$ whenever $\mathbf{A}$ is a PQA over $\mathbb{H}^{M \times N}$.
\par We now discuss an example of a $2\times 2$ quaternary PCP, and provide a list of quaternary PCPs that can be derived from PQAs using the proposed procedure. For the PQA in \eqref{eq:example_PQA}, the matrices $\tilde{\mathbf{D}}_{\mathrm{h}}={\mathbf{D}}_{\mathrm{h}}+{\mathbf{D}}_{\mathrm{v}}$ and $\tilde{\mathbf{D}}_{\mathrm{v}}={\mathbf{D}}_{\mathrm{h}}-{\mathbf{D}}_{\mathrm{v}}$ are
\begin{equation}
\label{eq:d_tild_hv}
\tilde{\mathbf{D}}_{\mathrm{h}}=\begin{pmatrix}
1 & \imi \\
1 & \imi  
\end{pmatrix} \, \mathrm{and} \, \,
\tilde{\mathbf{D}}_{\mathrm{v}}=\begin{pmatrix}
1 & \imi \\
-1 & -\imi  
\end{pmatrix}.
\end{equation}
The pair in \eqref{eq:d_tild_hv} forms a PCP over $\mathcal{C}$. Our procedure can also be used to transform one-dimensional or multi-dimensional PQAs into PCPs of the same size. For instance, one-dimensional quaternary PCPs with lengths $4,6,8,10,14,16,18,26,30,38,42,50,54,62,74,82,90$ and $98$ can be derived using the perfect quaterion sequences in \cite{pqa_2}.  Prior work has constructed $1\mathrm{D}$ PCPs of these lengths using techniques that do not use PQAs. Our PQA-based construction is useful in designing new multi-dimensional PCPs which have not been discovered in the literature. For instance, new quaternary PCP matrices of size $2^{n} \times 2^{n}$, and quaternary PCP tensors of size $2^{n} \times 2^{n} \times 2^{n} \times 2^{n}$ can be constructed for $2\leq n \leq 6$ from the PQAs in \cite{pqa_2}. The PQA constructions in \cite{pqa_2} are based on recursive algorithms or exhaustive search over a class of functions to generate such arrays. It is important to note that the periodic autocorrelation for the tensor case is multi-dimensional. For instance, the periodic autocorrelation of a 4D-array $\boldsymbol{\mathcal{A}}\in \mathbb{Q}^{M\times N \times S \times T}$ is defined as $\boldsymbol{\mathcal{R}}_{\mathcal{A}}$ where ${\mathcal{R}}_{\boldsymbol{\mathcal{A}}}(m,n,s,t)=\sum_{k,\ell,u,v}\mathcal{A}(k,\ell,u,v)\mathcal{A}^{\ast}(\langle k+m \rangle_M,\langle \ell+n \rangle_N, \langle u+s \rangle_S ,\langle v+t \rangle_T)$. Our procedure can also be used to decompose the PQAs in \cite{pqa_1}, \cite{pqa_3}, and \cite{pqa_4} into PCPs. 
\par Now, we focus on quaternary PCP matrices and study their complementary property using the 2D-discrete Fourier transform (2D-DFT). For $\mathbf{X} \in \mathbb{C}^{M \times N}$, we define $\mathcal{F}_{\mathrm{2D}}(\mathbf{X})$ as the 2D-DFT of $\mathbf{X}$. For example, $\mathcal{F}_{\mathrm{2D}}(\boldsymbol{\delta})=\mathbf{1}$. An interesting property of the 2D-DFT is that $\mathcal{F}_{\mathrm{2D}}(\mathbf{R}_{\mathbf{X}})=|\mathcal{F}_{\mathrm{2D}}(\mathbf{X})|^2$ \cite{2D_DFT}. For a PCP $\tilde{\mathbf{A}}_{\mathrm{h}},\tilde{\mathbf{A}}_{\mathrm{v}}\in \mathcal{C}^{M \times N}$, applying 2D-DFT on both sides of $\mathbf{R}_{\tilde{\mathbf{A}}_{\mathrm{h}}}+\mathbf{R}_{\tilde{\mathbf{A}}_{\mathrm{v}}}=2MN \boldsymbol{\delta}$ results in \cite{pcp_bomer}
\begin{equation}
\label{eq:fourier}
|\mathcal{F}_{\mathrm{2D}}(\tilde{\mathbf{A}}_{\mathrm{h}})|^2+|\mathcal{F}_{\mathrm{2D}}(\tilde{\mathbf{A}}_{\mathrm{v}})|^2 = 2MN \mathbf{1}.
\end{equation}
From a beamforming perspective, $|\mathcal{F}_{\mathrm{2D}}(\tilde{\mathbf{A}}_{\mathrm{h}})|^2$ is the power of the discrete beam pattern generated when $\tilde{\mathbf{A}}_{\mathrm{h}}$ is applied to a planar antenna array \cite{quasiomni}. When $\tilde{\mathbf{A}}_{\mathrm{h}}$ and $\tilde{\mathbf{A}}_{\mathrm{v}}$ are applied along the orthogonal polarizations of a dual polarized beamforming (DPBF) system, it can be observed from \eqref{eq:fourier} that the sum of the beam powers taken across both polarizations is constant at all the discrete beam pattern locations. As a result, PCPs result in quasi-omnidirectional beams when applied to DPBF systems. The quaternary nature of the PCPs derived in this paper allows their application to DPBF systems with just two-bit phase shifters. 
\begin{figure}[h!]
\centering
\subfloat[ $|\mathcal{F}_{\mathrm{2D}}(\tilde{\mathbf{A}}_{\mathrm{h}})|^2$]{\includegraphics[trim=1.5cm 6.5cm 2cm 7cm,clip=true,width=4.25cm, height=4.25cm]{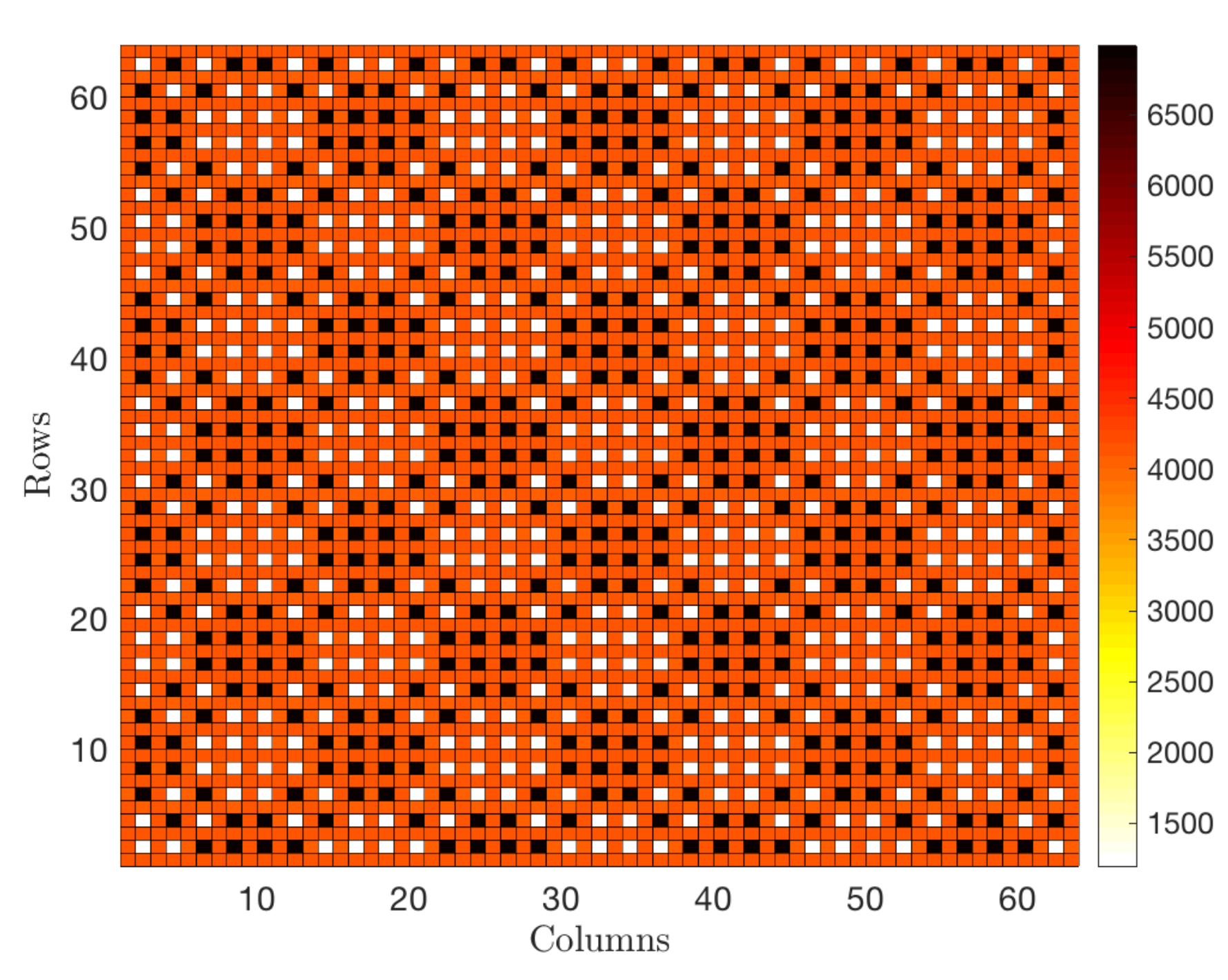}\label{fig:spec_h}}
\:\:\:
\subfloat[$|\mathcal{F}_{\mathrm{2D}}(\tilde{\mathbf{A}}_{\mathrm{v}})|^2$]{\includegraphics[trim=1.5cm 6.5cm 2cm 7cm,clip=true,width=4.25cm, height=4.25cm]{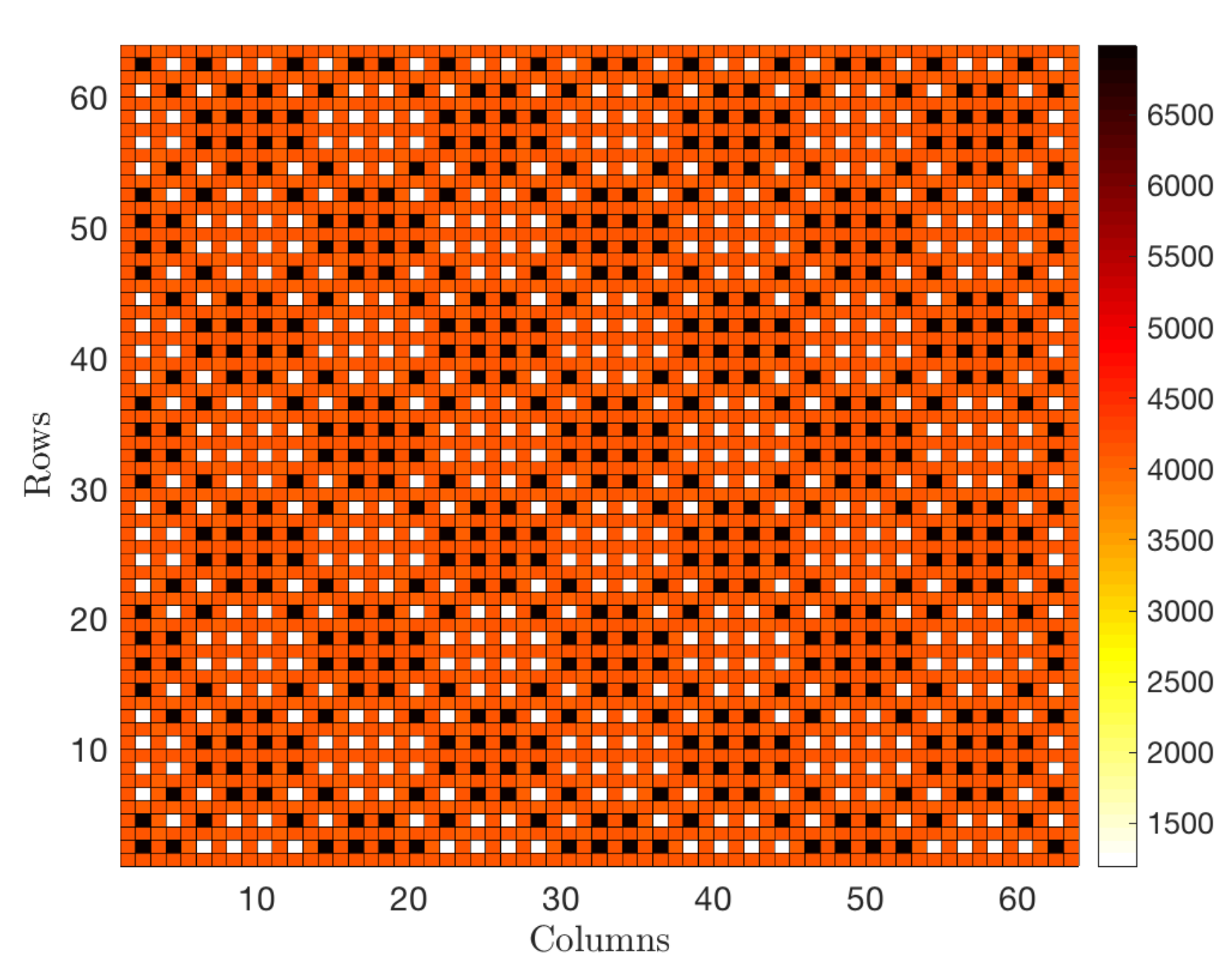}\label{fig:spec_v}}
\caption{The squared 2D-DFT magnitudes of the quaternary PCPs constructed with our method by exploiting the right periodic autocorrelation of the $64 \times 64$ PQA in \cite{pqa_2}. The sum of these matrices is $2\times 64^2$ at all the locations as per \eqref{eq:fourier}. \small \normalsize}
\end{figure}
\par For $4\times 4$, $8 \times 8$, $16 \times 16$, and $32 \times 32$ PCPs derived from the PQAs in \cite{pqa_2}, we observed that $|\mathcal{F}_{\mathrm{2D}}(\tilde{\mathbf{A}}_{\mathrm{h}})|^2=MN \mathbf{1}$. Equivalently, $\mathbf{R}_{\tilde{\mathbf{A}}_{\mathrm{h}}}=MN\boldsymbol{\delta}$. In such a case, the PCP is a collection of two perfect quaternary arrays \cite{pquat} as $\tilde{\mathbf{A}}_{\mathrm{h}}$ and $\tilde{\mathbf{A}}_{\mathrm{v}}$ have a perfect periodic autocorrelation. These arrays are different from the perfect quaternary arrays used in \cite{pqa_4} to construct PQAs. The construction in \cite{pqa_4} is based on an inflation technique which transforms a $M \times N$ perfect quaternary array into an $Md\times Nd$ PQA, where $d=MN-1$ is prime. Our method generates PCPs which have the same size of the underlying PQA and is not the same as the reverse construction in \cite{pqa_4}.
\par The proposed construction does not always result in PCPs which contain two perfect quaternary arrays. For example, it can be observed that $|\mathcal{F}_{2\mathrm{D}}({\tilde{\mathbf{D}}_{\mathrm{h}}})|^2/4\neq \mathbf{1}$ and $|\mathcal{F}_{2\mathrm{D}}({\tilde{\mathbf{D}}_{\mathrm{v}}})|^2/4 \neq \mathbf{1}$ for the PCP in \eqref{eq:d_tild_hv}. In this case, ${\tilde{\mathbf{D}}_{\mathrm{h}}}$ and ${\tilde{\mathbf{D}}_{\mathrm{v}}}$ are not perfect quaternary arrays although they form a PCP. Another example of a non-trivial PCP is one that is generated from a $64\times 64$ PQA. The squared 2D-DFT magnitudes of the $64 \times 64$ matrices in this PCP are shown in Fig. 1. As the 2D-DFT magnitudes in Fig. 1 vary across the entries, the matrices in this PCP are not perfect quaternary arrays.
\section{Construction of quaternary PCPs using the DFT and the left periodic autocorrelation}\label{sec:leftperiodic}
In this section, we use the properties of the DFT to obtain a class of quaternary PCPs from the PCPs constructed in Section \ref{sec:construct}. We also derive a new class of quaternary PCPs by exploiting the delta structure of the left periodic autocorrelation of PQAs.
\par We now discuss how the PCP property is preserved with unitary scalar multiplication, 2D-circulant shifting, conjugation, and flipping. To show this invariance, we first note that every matrix pair $\tilde{\mathbf{A}}_{\mathrm{h}}, \tilde{\mathbf{A}}_{\mathrm{v}} \in \mathcal{C}^{M\times N}$ which satisfies the 2D-DFT equation in \eqref{eq:fourier} is a PCP. From \eqref{eq:fourier}, we observe that $c_1 \tilde{\mathbf{A}}_{\mathrm{h}}, c_2 \tilde{\mathbf{A}}_{\mathrm{v}}$ is a PCP when $|c_1|=1$ and $|c_2|=1$. Next, we define an $(r,t)$ 2D-circulant shift of $\mathbf{A}$ as $\mathcal{S}_{r, t}(\mathbf{A})$. The matrix $\mathcal{S}_{r, t}(\mathbf{A})$ is obtained by circulantly shifting every column of $\mathbf{A}$ by $r$ units down, and circulantly shifting every row of the resultant by $t$ units to the right.  An interesting property of the 2D-DFT is that $|\mathcal{F}_{2\mathrm{D}}\left(\mathcal{S}_{r, t}(\mathbf{A})\right)|=|\mathcal{F}_{2\mathrm{D}}(\mathbf{A})|$. By using this equivalence in \eqref{eq:fourier}, it follows that $\mathcal{S}_{r_1, t_1}(\tilde{\mathbf{A}}_{\mathrm{h}}), \mathcal{S}_{r_2, t_2}(\tilde{\mathbf{A}}_{\mathrm{v}})$ form a PCP for any arbitrary integer pairs $(r_1,t_1)$ and $(r_2,t_2)$. Using the 2D-DFT properties \cite{image_process}, it can also be shown that the conjugates of the matrices in a PCP, i.e., $\tilde{\mathbf{A}}^{\ast}_{\mathrm{h}}$ and $ \tilde{\mathbf{A}}^{\ast}_{\mathrm{v}}$, form a PCP. Furthermore, the matrices $\tilde{\mathbf{A}}_{\mathrm{h},\mathrm{flip}}$ and $\tilde{\mathbf{A}}_{\mathrm{v},\mathrm{flip}}$ form a PCP as their 2D-DFT magnitudes are just the flipped versions of $|\mathcal{F}_{2\mathrm{D}}(\tilde{\mathbf{A}}_{\mathrm{h}})|$ and $|\mathcal{F}_{2\mathrm{D}}(\tilde{\mathbf{A}}_{\mathrm{v}})|$ \cite{image_process}. These invariance laws also hold for multi-dimensional PCPs by the multi-dimensional DFT properties.
\par We define the $2\mathrm{D}$-left periodic autocorrelation of $\mathbf{A}\in \mathbb{Q}^{M \times N}$ as an $M \times N$ matrix $\mathbf{L}_{\mathbf{A}}$ such that \cite{pqa_1}
\begin{equation}
L_{\mathbf{A}}(m,n)=\sum^{M-1}_{k=0} \sum^{N-1}_{\ell=0} A^{\ast}(\langle k+m \rangle_M,\langle \ell+n \rangle _N)A(k,\ell).
\end{equation}
The variables $k$ and $\ell$ in the summation can be replaced by $k-m$ and $\ell-n$ to write
\begin{equation}
\label{eq:left_exp}
L_{\mathbf{A}}(m,n)=\sum^{M-1}_{k=0} \sum^{N-1}_{\ell=0} A^{\ast}( k, \ell)A(\langle k-m \rangle_M,\langle \ell-n \rangle _N).
\end{equation}
When $\mathbf{A}$ is a PQA, its left periodic autocorrelation $\mathbf{L}_{\mathbf{A}}$ is a delta function, i.e., $L_{\mathbf{A}}(m,n)=0\, \forall \,(m,n)\neq (0,0)$ \cite{pqa_1}. From \eqref{eq:ACR_quatn} and \eqref{eq:left_exp}, we observe that $L_{\mathbf{A}}(m,n)=R_{\mathbf{A}^{\ast}}(-m,-n)$. We put these observations together to conclude that $R_{\mathbf{A}^{\ast}}(m,n)=0\, \forall \,(m,n)\neq (0,0)$. Equivalently, $\mathbf{A}^{\ast}$ is a PQA by the definition in \eqref{eq:PQA_def}. The delta structure of the left periodic autocorrelation of a PQA allows the use of our PCP construction in Section \ref{sec:construct} over the complex conjugate of a PQA.
\par Our PCP construction first decomposes a PQA into two complex components according to \eqref{eq:quat_mat_dec}. The PQA $\mathbf{A}^{\ast}$ can be expressed as $(\mathbf{A}_{\mathrm{h}}+\mathbf{A}_{\mathrm{v}}\imj )^{\ast}=\mathbf{A}^{\ast}_{\mathrm{h}}+\imj^{\ast}\mathbf{A}^{\ast}_{\mathrm{v}}$. The term $\imj^{\ast}\mathbf{A}^{\ast}_{\mathrm{v}}$ can be simplified to $-\mathbf{A}_{\mathrm{v}}\imj$ using the property that $x \imj^{\ast} y = -x y^{\ast} \imj$ for $x \in \mathbb{C}$ and $y\in \mathbb{C}$ [Proof in the Appendix]. As a result, $\mathbf{A}^{\ast}=\mathbf{A}^{\ast}_{\mathrm{h}}-\mathbf{A}_{\mathrm{v}}\imj$. Now, it can be shown using Theorem \ref{th:main} that $\mathbf{A}^{\ast}_{\mathrm{h}}$ and $-\mathbf{A}_{\mathrm{v}}$ form a PCP. Similar to the construction in \eqref{eq:tildes}, we derive PCP matrices $\hat{\mathbf{A}}_{\mathrm{h}}$ and $\hat{\mathbf{A}}_{\mathrm{v}}$. These matrices can be expressed as 
\begin{align}
\nonumber
\hat{\mathbf{A}}_{\mathrm{h}}&=\mathbf{A}^{\ast}_{\mathrm{h}}-\mathbf{A}_{\mathrm{v}}\,\, \mathrm{and}\\
\label{eq:hat_leftpcp}
\hat{\mathbf{A}}_{\mathrm{v}}&=\mathbf{A}^{\ast}_{\mathrm{h}}+\mathbf{A}_{\mathrm{v}}.
\end{align}
In summary, our method exploits the delta structure of the left periodic autocorrelation of a PQA matrix $\mathbf{A}$ to construct the PCP $\hat{\mathbf{A}}_{\mathrm{h}},\hat{\mathbf{A}}_{\mathrm{v}}$. In Fig. 2, we show the squared 2D-DFT magnitudes of $\hat{\mathbf{A}}_{\mathrm{h}}$ and $\hat{\mathbf{A}}_{\mathrm{v}}$ corresponding to the $64 \times 64$ PQA in \cite{pqa_2}.  We claim that such a PCP cannot be derived from $\tilde{\mathbf{A}}_{\mathrm{h}},\tilde{\mathbf{A}}_{\mathrm{v}}$ by applying transformations such as unitary scalar multiplication, 2D-circulant shifting, conjugation, and flipping. These transformations either preserve or flip the 2D-DFT magnitude of a matrix, which is not the case with the PCPs shown in Fig. 1 and Fig. 2. Therefore, our method allows decomposing any PQA in $\mathbb{H}^{M \times N}$ into two distinct pairs of quaternary PCPs that are defined by \eqref{eq:tildes} and \eqref{eq:hat_leftpcp}.
\begin{figure}[h!]
\centering
\subfloat[ $|\mathcal{F}_{\mathrm{2D}}(\hat{\mathbf{A}}_{\mathrm{h}})|^2$]{\includegraphics[trim=1.5cm 6.5cm 2cm 7cm,clip=true,width=4.25cm, height=4.25cm]{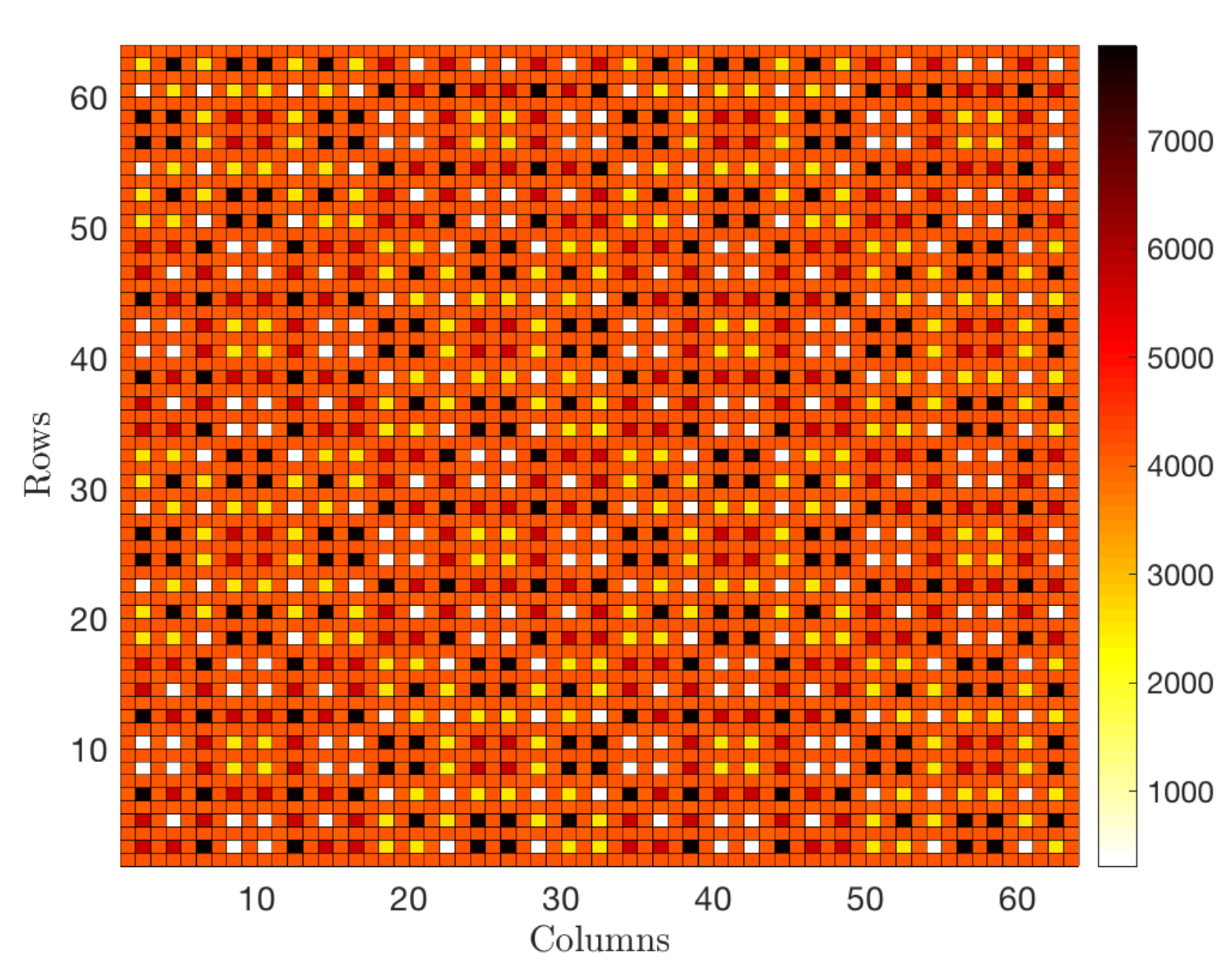}\label{fig:spec_h_hat}}
\:\:\:
\subfloat[$|\mathcal{F}_{\mathrm{2D}}(\hat{\mathbf{A}}_{\mathrm{v}})|^2$]{\includegraphics[trim=1.5cm 6.5cm 2cm 7cm,clip=true,width=4.25cm, height=4.25cm]{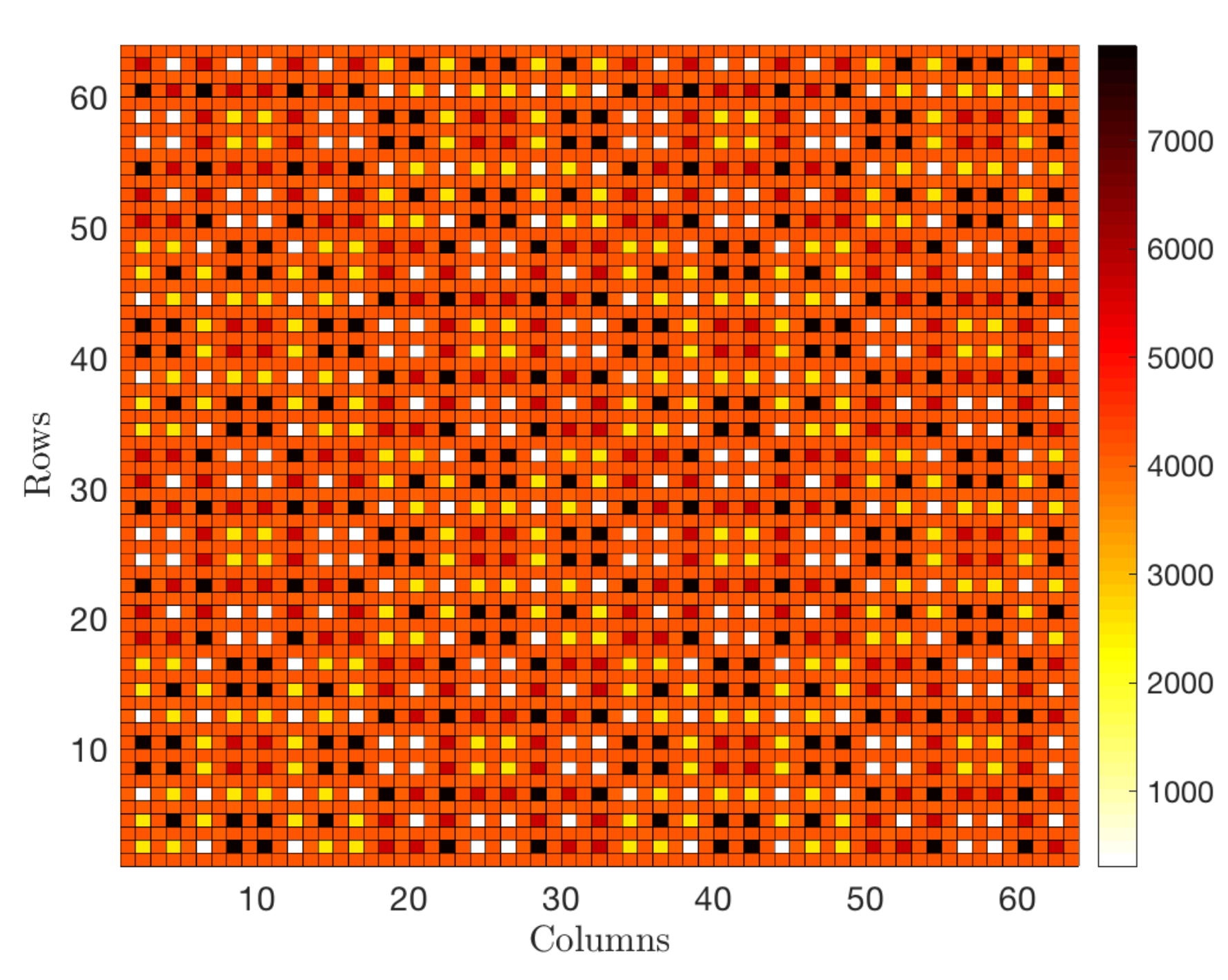}\label{fig:spec_v_hat}}
\caption{ The squared 2D-DFT magnitudes of the quaternary PCPs constructed with our method by exploiting the left periodic autocorrelation of the $64 \times 64$ PQA in \cite{pqa_2}. The sum of these matrices is $2\times 64^2$ at all the locations as per \eqref{eq:fourier}. \small \normalsize}
\end{figure}
\par We would like to mention that the main focus of this paper is on constructing PCPs from PQAs. An interesting question that arises is if PQAs can be derived from PCPs using the reverse of the proposed construction. To answer this question, we consider a PCP $\mathbf{B}_{\mathrm{h}}, \mathbf{B}_{\mathrm{v}} \in \mathbb{C}^{M\times N}$ such that $\Vert \mathbf{B}_{\mathrm{h}} \Vert_{\mathrm{F}}= \sqrt{MN}$ and $\Vert \mathbf{B}_{\mathrm{v}} \Vert_{\mathrm{F}}=\sqrt{MN}$. By definition, $\mathbf{R}_{\mathbf{B}_{\mathrm{h}}}+\mathbf{R}_{\mathbf{B}_{\mathrm{v}}}=2MN \boldsymbol{\delta}$. From Lemma \ref{lem:decomp}, it can be concluded that the quaternion matrix $\mathbf{B}=\mathbf{B}_{\mathrm{h}}+\mathbf{B}_{\mathrm{v}}\imj$ is perfect when $ \mathbf{B}_{\mathrm{h}} \star  \mathbf{B}_{\mathrm{v}}=\mathbf{B}_{\mathrm{v}} \star  \mathbf{B}_{\mathrm{h}}$. Therefore, PCPs which are commutative with conjugate-free cross correlation can be used to construct PQAs.
\par Now, we consider the right periodic autocorrelation and identify sufficient conditions to construct PQAs in $\mathbb{H}^{M \times N}$ from quaternary PCPs. For a quaternary PCP $\mathbf{B}_{\mathrm{h}}, \mathbf{B}_{\mathrm{v}} \in \mathcal{C}^{M\times N}$, we define $\mathbf{A}_{\mathrm{h}}=(\mathbf{B}_{\mathrm{h}}+\mathbf{B}_{\mathrm{v}})/2$ and $\mathbf{A}_{\mathrm{v}}=(\mathbf{B}_{\mathrm{h}}-\mathbf{B}_{\mathrm{v}})/2$. Now, it can be shown that $\mathbf{A}=\mathbf{A}_{\mathrm{h}}+\mathbf{A}_{\mathrm{v}}\imj$  is a PQA when $ \mathbf{B}_{\mathrm{h}} \star  \mathbf{B}_{\mathrm{v}}=\mathbf{B}_{\mathrm{v}} \star  \mathbf{B}_{\mathrm{h}}$. Furthermore, all the entries of $\mathbf{A}$ are elements in $\mathbb{H}$ only when $A_{\mathrm{h}}(k,\ell)A_{\mathrm{v}}(k,\ell)=0$ for every $k, \ell$. This condition translates to $B_{\mathrm{h}}(k,\ell)=\pm B_{\mathrm{v}}(k,\ell)$ for every $k, \ell$. In conclusion, the reverse of our construction allows mapping a quaternary PCP $\mathbf{B}_{\mathrm{h}}, \mathbf{B}_{\mathrm{v}}$ to a PQA in $\mathbb{H}^{M \times N}$ if the PCP satisfies the following properties:
\begin{align*}
(a)&  \,\, \mathbf{R}_{\mathbf{B}_{\mathrm{h}}}+\mathbf{R}_{\mathbf{B}_{\mathrm{v}}}=2MN \boldsymbol{\delta},\,\,  \\
(b)& \,\, \mathbf{B}_{\mathrm{h}} \star \mathbf{B}_{\mathrm{v}}= \mathbf{B}_{\mathrm{v}} \star \mathbf{B}_{\mathrm{h}}, \, \mathrm{and}\\
(c)& \,\, B_{\mathrm{h}}(k,\ell)=\pm B_{\mathrm{v}}(k,\ell)\,\, \forall k, \ell.
\end{align*}
To the best of our knowledge, the conditions $(a)-(c)$ have not been presented in prior work. Although $(a)-(c)$ are derived by exploiting the right periodic autocorrelation of a PQA, it can be shown that use of left periodic autocorrelation also results in the same set of conditions. These conditions are only sufficient and not necessary as there may be other methods to construct PQAs from PCPs. We believe that the conditions in $(a)-(c)$ can provide new insights into constructing PQAs over $\mathbb{H}$. 
\section{Conclusions and future work}
In this paper, we established a connection between perfect arrays over quaternions and periodic complementary arrays over complex numbers. We also demonstrated how perfect quaternion arrays can be transformed to quaternary periodic complementary pairs. Finally, we identified sufficient conditions to construct perfect quaternion arrays over the basic unit quaternions from quaternary periodic complementary pairs. In future work, we will study the use of perfect quaternion arrays for beamforming in low resolution phased arrays. 
\balance
\section*{Appendix}
We first prove that $x\imj y=xy^{\ast} \imj$ when $x,y \in \mathbb{C}$. Using the representation in \eqref{eq:quat_split}, $x\imj y$ can be written as 
\begin{align}
\nonumber
x\imj y&=(x_1+x_2 \imi) \imj (y_1+y_2 \imi)\\
\nonumber
&=(x_1+x_2 \imi) (y_1\imj -y_2 \imk)\\
\nonumber
&=(x_1 y_1 + x_2 y_2) \imj +(x_2 y_1-x_1 y_2) \imk \\
\label{eq:xjy}
&=[(x_1 y_1 + x_2 y_2) +(x_2 y_1-x_1 y_2) \imi] \imj. 
\end{align}
It can be observed that the right hand side of \eqref{eq:xjy} is $xy^{\ast} \imj$. Using $\imj^{\ast}=-\imj$ and the result in \eqref{eq:xjy}, it can be shown that $x \imj ^{\ast}y=-xy^{\ast} \imj$ for any $x, y \in \mathbb{C}$.
\bibliographystyle{IEEEtran}
\bibliography{refs}
\end{document}